\documentclass{article}
\usepackage{graphicx} 
\usepackage[utf8]{inputenc}
\usepackage{pgf,tikz,pgfplots}
\pgfplotsset{compat=1.17}
\usetikzlibrary{arrows}
\usetikzlibrary{fit}
\usepackage{amsthm}
\usepackage{array}
\usepackage{tabularx}
\newcolumntype{Y}{>{\centering\arraybackslash}X}
\usepackage{cancel}
\usepackage{bm}
\usepackage{breqn}
\usepackage[utf8]{inputenc}
\usepackage{csquotes}
\expandafter\let\csname equation*\endcsname\relax
\expandafter\let\csname endequation*\endcsname\relax
\usepackage{amsmath, amssymb}
\usepackage{amsfonts}
\usepackage{amstext}
\usepackage{physics}
\usepackage{color}
\usepackage{enumerate}
\usepackage{graphicx}
\usepackage{float}
\usepackage{setspace} 
\usepackage{lipsum}
\usepackage{authblk}

\date{} 

\usepackage{dcolumn}
\usepackage{caption}
\usepackage{subcaption}
\usepackage{epsfig}
\usepackage{epsf}
\usepackage{epstopdf}
\usepackage{macros}

\usepackage{enumitem}
\usepackage[numbers,sort&compress]{natbib}
\usepackage{hyperref}
\usepackage{xcolor}
\usepackage{tikz}
\usepackage{pgfplots}
\def\nn{\nonumber}
\usepackage{esint}

\usepackage{scalerel}
\usepackage{stackengine}
\usepackage{braket}
\usepackage{titlesec}
\usepackage{verbatim}
\usepackage{tikz}
\newtheorem{lemma}[theorem]{Lemma}

\newcommand{\vect}[1]{\boldsymbol{#1}}
\DeclareSymbolFont{usualmathcal}{OMS}{cmsy}{m}{n}
\DeclareSymbolFontAlphabet{\mathcal}{usualmathcal}

\newcommand{\be}{\begin{equation}}
\newcommand{\ee}{\end{equation}}
\newcommand{\bea}{\begin{eqnarray}}
\newcommand{\eea}{\end{eqnarray}}
\usepackage{ulem}

\usepackage{color}
\definecolor{dgreen}{rgb}{0,0.7,0}

\title{Absence of {hidden} analytic conserved quantities in harmonically confined rods}
\author[1,2]{Sahil Kumar Singh}
\author[2]{Abhishek Dhar}
\author[3,4]{Sanjay Moudgalya}

\affil[1]{Complex Systems and Statistical Mechanics, Department of Physics and Materials Science,
University of Luxembourg, 30 Avenue des Hauts-Fourneaux, L-4362 Esch-sur-Alzette, Luxembourg}

\affil[2]{International Centre for Theoretical Sciences, Tata Institute of Fundamental Research, Bangalore 560089, India}

\affil[3]{School of Natural Sciences, Technische Universit\"at M\"unchen (TUM), James-Franck-Str. 1, 85748 Garching, Germany}
\affil[4]{Munich Center for Quantum Science and Technology (MCQST), Schellingstr. 4, 80799 M\"unchen, Germany}

\begin{document}

\maketitle
\begin{abstract}
    Systems of hard rods of equal length in a one-dimensional harmonic trap have been observed to exhibit peculiar non-ergodic behavior that might suggest the existence of a novel hidden conserved quantity beyond the two well known ones, i.e., the total energy and the center-of-mass energy.
    In this work, we investigate this possibility by systematically constraining the forms of the conserved quantities, and we rigorously rule out the existence of any extra hidden conserved quantity that is analytic in the positions and momenta of the rods involved.
    We do so by showing two key results: conservation during free motion demands the $U(1)$ invariance of these quantities under rotations of the position and momenta of each rod, and conservation during collisions demand an $S_N$ invariance under the permutation of the momenta of the rods as long as one of the rods have non-zero length. 
    We then show that these conditions imply that any conserved quantity is functionally dependent on the two known conserved quantities. 
    In addition, we show that in the special case where all rods have zero length (i.e., when they are point particles), conservation under collisions only requires invariance under a smaller $S_N$ group of permutations of the labels of the rods, which leads to a much larger set of analytic conserved quantities that we explicitly write down.
    In all, this rigorously clarifies the structure of conserved quantities in the hard rod problem, and motivates the application of such systematic methods to other classical systems.  
\end{abstract}
\tableofcontents
\section{Introduction}
Large-scale behavior of {the dynamics of complex} many-particle systems at {late times} can often be understood using simple universal laws that depend solely on the symmetries {and conserved quantities of the system, } with microscopic details becoming irrelevant~\cite{RevModPhys.71.S346,goldenfeld1992phase}. 
{A classic example is the emergence of the theory of} thermodynamics {which} relies only on the symmetries and conservation laws of the underlying system~\cite{Callen:450289} {but nevertheless is known to almost ubiquitously apply to many-body systems encountered in much of physics}.
{This leads to the broad classification of systems in terms of their conserved quantities.
In particular a clear distinction is made between two classes: \textit{Integrable} systems,  which have as many independent conserved quantities in involution as they have degrees of freedom  and allow a greater degree of microscopic analytical tractability of the system~\cite{babelon2003introduction,arutyunov2019elements,perelomov1990integrable}, and \textit{Non-Integrable} systems, which have just a few conserved quantities and whose microscopic dynamics is generally analytically intractable.
 Many of the simplifications in the standard theories of thermodynamics only apply to non-integrable systems, which is enabled by the fact that their dynamics is ergodic in the phase space allowed by the few conserved quantities~\cite{Gallavotti2016Ergodicity,bookvulpiani}, {such as the relaxation of late-time observables to values predicted by Gibbs ensembles.}
In contrast, even though integrable systems admit microscopic solvability, their extensively many conserved quantities can lead to rather intricate late-time dynamics, {and the relaxation of late-time observables is not always guaranteed even though some do relax to values predicted by the more complicated} Generalized Gibbs Ensemble (GGE) that accounts for all constants of motion~\cite{rigol2007,vidmar2016,PhysRevX.6.041065,PhysRevLett.117.207201,10.21468/SciPostPhysLectNotes.18,Spohn1991LargeScale,Spohn2020GGEToda,Doyon2019GeneralizedHydrodynamicsToda}.}
 %

%

\par
As much as one would like a clear distinction between integrable and non-integrable systems, there are examples of systems that appear to be neither, in that their dynamics does not exhibit ergodicity in the phase space allowed by the known conserved quantities, nor can any extra conserved quantities be immediately identified.
A classic example is the Fermi-Pasta-Ulam-Tsingou (FPUT) numerical experiment~\cite{fermi,Dauxois2008FPUTsingou}, which found a {prolonged} lack of equipartition and non-ergodic behavior {on the timescales accessible in the original simulations,} despite the presence of an integrability-breaking term. {Subsequent studies showed that equipartition can eventually be reached, but on timescales that become very large as the nonlinear perturbation is weakened~\cite{benettin2005equipartition,ponno2011twostage,onorato2015route}.}
This led to a decades-long effort to understand the subtle nature of thermalization, ergodicity, integrability, and metastability in nonlinear many-body systems, with numerous works performing numerical experiments on the effect of integrability breaking terms on thermalization~\cite{,carati2005challenge,gallavotti2008statusreport,benettin2011timescales,DiCintio2018PinnedToda,Dhar2019PinnedToda,Fu_2019,PhysRevE.100.032217,PhysRevLett.122.054102,PhysRevE.104.014218,PhysRevLett.128.134102,vanovac2026weakintegrabilitybreakingperturbations,PhysRevLett.125.240604,Durnin_2021,PhysRevResearch.6.023083}.
One such system that we will focus on here is a system of hard rods (with non-zero lengths) in a harmonic trap~\cite{bulchandani,bagchi}.
While hard rods without any confinement as well as hard rods with zero lengths in a harmonic trap are both known to be integrable, the combination is expected to break integrability.
One obvious conserved quantity of such a system is the total energy, and a slightly less obvious one is the center-of-mass energy, hence the system is expected to be ergodic in the phase space allowed by these conserved quantities.
Nevertheless, previous works found evidence for a lack of ergodicity in this system in the Poincare sections of the phase space as well as zero Lyapunov exponents for the $3$-{rod} system that signal the absence of classical chaos~\cite{bulchandani,bagchi}. 
{In addition,  even observables in systems with a large number of large number of rods were observed to not thermalize to their expected Gibbs ensembles{~\cite{bulchandani,bagchi,b587-8yyt}},  in spite of their Lyapunov exponents being non-zero.} 
This begs the question of whether there are hidden integrals of motion that give rise to this apparent non-ergodicity, lack of thermalization and regular orbits in phase space.
Similar conundrums recently arose in quantum many-body systems, which, similar to classical many-body systems, are universally expected to undergo a quantum version of thermalization with restrictions imposed by their conserved quantities, unless they are integrable~\cite{Rigol2008Thermalization,DAlessio03052016,Mori_2018,gring2012relaxation,mallayya2019prethermalization,bertini2015prethermalization,surace2023weakintegrability,lin2017quasiparticle,lin2017explicit,lin2020slowthermalization,li2021hierarchy,long2023phenomenology,luitz2020prethermalization}.
However, the discovery of quantum ergodicity-breaking phenomena known as quantum many-body scars and Hilbert space fragmentation put into question this simple dichotomy of quantum many-body systems~\cite{Moudgalya_2022,SerbynAbaninPapi_2021,annurev:/content/journals/10.1146/annurev-conmatphys-031620-101617}.
There it was ultimately found that the non-ergodicity in a large class of such systems can indeed be attributed to hidden conserved quantities of the system, which can be highly complicated and non-local, demanding systematic algebraic methods to expose them~\cite{PhysRevX.12.011050,PhysRevX.14.041069,PhysRevB.107.224312}, {similar to ones also known in certain classical many-body systems~\cite{Dhar1993,PhysRevLett.73.2135,MKHariMenon_1995,Menon_Barma_Dhar_1997}.}
In addition, systematic methods are being widely applied to rigorously demonstrate the absence of local conserved quantities in a large class of commonly-studied quantum many-body systems, essentially proving their non-integrability~\cite{Shiraishi_2019,shiraishi2024absence,8l6f-z1jm,Shiraishi2025IntegrabilityClassification,Shiraishi2026XYXYZ}.
Given the success of systematic methods in quantum many-body systems, in this work, we apply classical analogs of those ideas to address the question of hidden conserved quantities in this classical system of hard rods in a harmonic trap.

Of course, performing systematic searches for conserved quantities in classical systems is by no means a new idea.
Classical integrability is often established through the Hamilton--Jacobi formalism, where separability can produce integrals in involution~\cite{Arnold1989}, or through a Lax representation, whose isospectral evolution yields conserved spectral invariants~\cite{babelon2003introduction,Lax1968}. 
Such structures need not be evident a priori, however, and in their absence it is natural to search directly for additional conserved quantities under controlled assumptions.
Some of the earliest works date back to 1887, when in the context of 3-body problem in gravitation, Heinrich Bruns proved the absence of  additional algebraic integrals of motion (polynomials+rational functions of the rectangular coordinates)~\cite{Bruns1887,JulliardTosel2000BrunsTheorem}, and Poincaré later showed the absence of additional analytic integrals in the same setting for sufficiently small mass ratios~\cite{Poincare1890,hp1892mna}.
There has also been a large body of works that systematically constrain the forms of Hamiltonian systems with two degrees of freedom that possess additional hidden integrals of motion,  which are reviewed in~\cite{HIETARINTA198787,KatsuyaNakagawa_2001,Pucacco_Rosquist_2005,10.1063/1.1917311}. 
Related extensions to homogeneous polynomial potentials with three or more degrees of freedom have also been developed~\cite{Przybylska_2009}.
To our knowledge, such an exploration has not been performed for many-body systems in a harmonic trap, and in this work we fill in this gap by studying many closely related systems in a harmonic trap.
First, we start by characterizing all analytic conserved quantities of multiple non-interacting rods in free harmonic motion{, which do not undergo collisions.
We show that these should be invariant under a $U(1)$ transformation that rotates between the individual position and momentum of each rod, and the complete set of such invariants are generated by a simple set of quadratic invariants (Theorem~\ref{th:freemotion}){, which also leads to its superintegrability}.  
Second, we turn to the case of rods with all zero lengths that are moving and colliding in a harmonic trap. 
We show that the conserved quantities there are invariant under the same $U(1)$ as free motion, as well as under the {permutation group $S_N$ that permutes the labels of the particles (thus permuting their positions $\{x_i\}$ and momenta $\{p_i\}$ in the same way), which allows us to construct their complete set of analytic conserved quantities (Theorem~\ref{thm:a=0}) {and show that it is superintegrable}.
Third, we study a system of rods in a harmonic trap undergoing a Stochastic Momentum Exchange Dynamics (SMED).
We show that their conserved quantities are invariant under the permutation group $S_N^{(p)}$ of exchange of the momenta $\{p_i\}$ in addition to the $U(1)$ of free motion, which we then show rules out the existence of an extra independent conserved quantity apart from the two known ones -- total energy and center of mass energy (Theorem~\ref{thm:Pexchange}).
Finally, we study the dynamics of rods in a harmonic trap with at least one rod of non-zero length.
We show that conserved quantities there too should also be invariant under the permutation group $S_N^{(p)}$ in addition of the $U(1)$, same as in the case of SMED, hence ruling out additional hidden analytic conserved quantities (Lemma~\ref{lem:collisions}, Theorem~\ref{thm:collisions}).
These four cases are summarized in Table~\ref{tab:symmetry-summary}.
The remainder of the paper is organized as follows.
In Sec.~\ref{sec:model}, we review the model of harmonically confined hard rods. 
In Sec.~\ref{sec:freemotionU1}, we construct exhaustively all conserved quantities of the free harmonic motion.
In Sec.~\ref{sec:a=0}, we analyze the special case of vanishing rod length, where a large family of conserved quantities appears.
In Sec.~\ref{sec:SMED}, we study the SMED system and show that no additional conserved quantities arise there.
In Sec.~\ref{sec:absenceofthird}, we present our main analytical result: when there is at least one rod of non-zero length, no independent conserved quantity analytic in the phase-space variables exists beyond the total energy and center-of-mass energy.
Finally, in Sec.~\ref{sec:discussion}, we summarize our results and discuss their implications for ergodicity and thermalization in classical many-body systems.

\begin{table}[t]
\centering
\small
\renewcommand{\arraystretch}{1.25}
\begin{tabularx}{\textwidth}{@{}
C{0.25\textwidth}
Y
C{0.22\textwidth}
@{}}
\hline
\textbf{System} & \textbf{Invariant Group} & \textbf{Conserved quantities} \\
\hline
Free motion
&
$U(1)$: $(x_i,p_i) \rightarrow (x_i \cos\theta + p_i \sin\theta, -x_i \sin \theta + p_i \cos\theta)$
&
Thm.~\ref{th:freemotion}
\\
Point particles
&
$U(1)$; $S^{(x,p)}_N$: $(x_i, p_i) \rightarrow (x_{\sigma(i)}, p_{\sigma(i)})$
&
Thm.~\ref{thm:a=0}
\\
SMED \& Rods with at least one non-zero length
&
$U(1)$; $S_N^{(p)}$: $(x_i, p_i) \rightarrow (x_i, p_{\sigma(i)})$; $S_N^{(x)}$: $(x_i, p_i) \rightarrow (x_{\sigma(i)}, p_i)$
&
Thms.~\ref{thm:Pexchange}, \ref{thm:collisions}
\\
\hline
\end{tabularx}
\caption{{Summary of our results for various systems studied in this work. 
In all these systems, we show that conserved quantities analytic in $\{x_i\}$ and $\{p_i\}$ should be invariant under certain group transformations.
These restrictions then allow us to derive the exhaustive set of analytic conserved quantities in each of these systems. }}
\label{tab:symmetry-summary}
\end{table}

\section{{Review of the model and its properties}}
\label{sec:model}
\subsection{{Definition and conserved quantities}}
We consider the well-studied system of hard rods in a harmonic trap.
This consists of $N$ hard rods of lengths $\{a_i\}$ and equal masses $m$, which are confined in a harmonic trap with a trap frequency of {$\omega$}.
The Hamiltonian is given by:
\begin{equation}
    H = \sum_{i=1}^N \left( \frac{p_i^2}{2 {m}} + \frac{1}{2} m \omega^2 x_i^2 \right) + \sum_{i} U_{i,i+1},\;\;\;U_{i,i+1} = 
    \begin{cases}
        \infty & \text{if } |x_i - x_{i+1}| \leq \frac{a_i + a_{i+1}}{2}, \\
        0 & \text{otherwise}.
    \end{cases}
\label{eq:Hamiltonian}
\end{equation}
where $U(r)$ is the hard-core interaction between the rods, $\{x_i\}$ are the (time-dependent) positions of the center-of-mass of the rods ordered as $x_1 < x_2 < \cdots < x_N$, $\{p_i\}$ are their (time-dependent) momenta.
{Without loss of generality}, we set $m = \omega = 1$.
The dynamics of this system thus consists of two parts: free-motion, and collisions, which we assume to be elastic.
It is convenient to separate the free part of the Hamiltonian and write its equations of motion as
\begin{equation}
    H_{\rm free} = \frac{1}{2}\sum_{i = 1}^N{(p_i^2 + x_i^2)} \implies \;\;\;\dot x_i = \frac{\partial H}{\partial p_i} = p_i,\;\;\;\dot p_i = -\frac{\partial H}{\partial x_i} = -x_i,
\label{eq:freemotion}
\end{equation}
where $\dot A$ denotes the time-derivative of $A$.
On the other hand, collisions only happen when two rods touch, and the two rods $n$ and $n+1$ involved then instantaneously exchange their momenta due to elasticity, giving rise to the condition
\begin{equation}
    (p_{n}^+, p_{n+1}^+) = (p_{n+1}^-, p_n^-),\;\;\;\text{if}\;\;\;x_{n+1} - x_n = {b_{n}},\;\;\;{b_{n} \defn \frac{a_n + a_{n+1}}{2}}. 
\label{eq:collision}
\end{equation}
where $p^-$ and $p^+$ are the momenta of the rods before and after the collision. 
Any conserved quantity $Q(\{x_i\}, \{p_i\})$ of the Hamiltonian Eq.~(\ref{eq:Hamiltonian}) has to be individually conserved under both the free-motion Eq.~(\ref{eq:freemotion}) and the collisions Eq.~(\ref{eq:collision}).
Hence it should satisfy the following conditions
\begin{gather}
    \dot Q(\{x_i\}, \{p_i\}) = \sum_i{\left(\dot x_i \frac{\partial Q}{\partial x_i} + \dot p_i\frac{ \partial Q}{\partial p_i}\right)} = \sum_i{\left(p_i\frac{ \partial Q}{\partial x_i} - x_i\frac{\partial Q}{\partial p_i}\right)} = 0,\label{eq:freeconserved}\\
    Q(\{\cdots, x_n, x_{n+1} = x_n + {b_n}, \cdots\}, \{\cdots, p_n, p_{n+1}, \cdots\}) = Q(\{\cdots, x_n, x_{n+1} = x_n + {b_n}, \cdots\}, \{\cdots, p_{n+1}, p_{n}, \cdots\})\;\;\forall\ n.\label{eq:collisionconserved}
\end{gather}
Since the condition of Eq.~(\ref{eq:collisionconserved}) involves setting $x_{n + 1} = x_n + b_n$, and exchanging the momenta $p_n$ and $p_{n+1}$, it is convenient to introduce compact notations for those operations.
 Here we will denote these using operations $\mathcal{X}^{(b)}_{n}$ and $\mathcal{P}_{n}$ that act on functions as
\begin{align}
     \mathcal{X}^{(b)}_{n} [Q(\{\cdots, x_n, x_{n+1}, \cdots\}, \{p_i\})]  &= Q(\{\cdots, x_{n}, x_{n} + b, \cdots \}, \{p_i\})\nn \\
     \mathcal{P}_{n} [Q(\{x_i\}, \{\cdots, p_n, p_{n+1}, \cdots\})]  &= Q(\{x_i\}, \{\cdots, p_{n+1}, p_n, \cdots\}), 
\label{eq:superopdefn}
 \end{align}
 where the compositions of these operations will simply be denoted as $\mP_n \mX_n^{(b)}$.
 For example, in this notation, Eq.~(\ref{eq:collisionconserved}) for a single pair of rods $n$ and $n+1$ reads
 \begin{equation}
 \label{eq:collisioncondition}
\mathcal{X}_n^{({b_n})}[Q] = \mathcal{P}_n\mathcal{X}_n^{({b_n})}[Q].
 \end{equation}
Two conserved quantities that satisfy Eqs.~(\ref{eq:freeconserved}) and (\ref{eq:collisionconserved}) are well-known to be energy $E$ and the center-of-mass energy $E_{\rm cm}$, given by:
\begin{equation}
    E = \sum_{i = 1}^N{(x_i^2 + p_i^2)},\;\;\; E_{\rm cm} = \left( \sum_{i=1}^N x_i \right)^2 + \left(\sum_{i=1}^N p_i \right)^2,
\label{eq:traditionalconserved}
\end{equation}
where we have omitted overall factors for convenience.
It is easy to see that these satisfy Eqs.~(\ref{eq:freeconserved}) and (\ref{eq:collisionconserved}). 
It is convenient in many places below to work with complex phase-space coordinates defined as
\begin{equation}\label{eq:z_def}
(z_i, \bar z_i) = (x_i+\mathrm{i}p_i, x_i- \mathrm{i}p_i) \;\;\iff\;\;(x_i, p_i) = \left(\frac{z_i+\bar z_i}{2}, \frac{z_i- \bar z_i}{2\mathrm{i}}\right),\;\;\;\;\text{where}\;\; \mathrm{i}=\sqrt{-1}.
\end{equation}
In these coordinates, the known conserved quantities are
\begin{equation}
    E = \sum_{i=1}^N{\left|z_i\right|^2},\;\;\;E_{\rm cm} = \left|\sum_{i=1}^N{z_i}\right|^2. 
\label{eq:EEcmdefn}
\end{equation}
\subsection{Systematic search for analytic conserved quantities}
Our main aim in this work is to systematically explore if $E$ and $E_{\rm cm}$ are the \textit{only} conserved quantities in the system of hard rods, or if there are additional ones.
%
%
First, we observe that given any two conserved quantities $Q_\alpha$ and $Q_\beta$ that satisfy Eqs.~(\ref{eq:freeconserved}) and (\ref{eq:collisionconserved}), their products and complex linear combinations also satisfy them.
Hence these conserved quantities form an \textit{algebra} of the form
\begin{equation}
    \mathbb{C}[\{Q_\alpha\}] \defn \{ \sum_{n}{\sum_{\{\alpha_j\}}{c_{\alpha_1 \cdots \alpha_n} Q_{\alpha_1} \cdots Q_{\alpha_n}}},\;\;\;c_{\alpha_1 \cdots \alpha_n} \in \mathbb{C}\},
\end{equation}
where $\{Q_\alpha\}$ {are referred to as the generators of the algebra}. 
{Usually these generators are also required to all be functionally independent generators, but here we do not necessarily impose such a restriction.} When these generators are independent bounded-degree polynomials of $\{x_i\}$ and $\{p_i\}$, we can also refer to $\mathbb{C}[\{Q_\alpha\}]$ as a \textit{polynomial ring}. 
For the case of hard rods we discuss here, the main question we wish to answer is if the set of analytic conserved quantities $\mathcal Q \in \mathbb{C}[\{x_i\}, \{p_i\}]$ that satisfy Eqs.~(\ref{eq:freeconserved}) and (\ref{eq:collisionconserved}) is the algebra generated by $E$ and $E_{\rm cm}$:
\begin{equation}
    \mathcal Q \overset{?}{=} \mathbb{C}[E, E_{\rm cm}]. 
\end{equation}
In Sec.~\ref{sec:absenceofthird}, we show that this is true when at least a single rod has a non-zero length, i.e., $a_i \neq 0$ for some $i$,  which demonstrates that there are no independent analytic conserved quantities for this system apart from $E$ and $E_{\rm cm}$.
We do so by mapping it onto the case of the SMED process, for which we show the same in Sec.~\ref{sec:SMED}.
However, in Sec.~\ref{sec:a=0}, we find, perhaps as expected, that this is not the case when $a_i = 0$ for all $i$, i.e., for harmonically confined point particles, for which there are numerous additional conserved quantities for which we obtain simple closed forms.
Along the way, in Sec.~\ref{sec:freemotionU1} we also recover the exhaustive set of conserved quantities that are conserved only during free motion, i.e., those that satisfy Eq.~(\ref{eq:freeconserved}), which is of course very well-known, but to our knowledge never shown formally {in the literature}.  
{In some cases, we also discuss the independence of any set of conserved quantities $\{Q_{\alpha_j},\;\;j = 1, \cdots, k\}$.}
This can be easily probed by computing the rank of their Jacobian matrix, which is of the form
\begin{equation}
\begin{pmatrix}
    \frac{\partial Q_{\alpha_1}}{\partial x_1}&\cdots &\frac{\partial Q_{\alpha_1}}{\partial x_N}& \frac{\partial Q_{\alpha_1}}{\partial p_1} & \cdots &\frac{\partial Q_{\alpha_1}}{\partial p_N}\\\vdots &\ddots & \vdots & \vdots & \ddots &
    \vdots\\
    \frac{\partial Q_{\alpha_k}}{\partial x_1}&\cdots &\frac{\partial Q_{\alpha_k}}{\partial x_N}& \frac{\partial Q_{\alpha_k}}{\partial p_1} & \cdots& \frac{\partial Q_{\alpha_k}}{\partial p_N}
\label{eq:jacobian}
\end{pmatrix},
\end{equation}
which implies their functional independence if it has rank $k$ for generic values of the variables.
While we have shown the Jacobian in the standard phase space coordinates, this can also be computed in other variables {that can be related to the phase space variables by a locally invertible transformation}, which might sometimes be more convenient, {as we will show in some examples below}.
Before we proceed, we note that our proof holds for any quantity which can be represented as a power series in {any open region around} the origin of the phase space, given by $\{x_i = 0\}$ and $\{p_i = 0\}$.}
Hence {it applies to all quantities that are} \textit{analytic} at the origin{, even though they might not be analytic elsewhere in the phase space}.
{This class of quantities is more general than \textit{globally analytic} quantities, which are analytic everywhere in the phase space.}
{We emphasize that, for nonzero contact distances, the origin does not belong to the physical hard-rod phase space. The analyticity considered here refers to an analytic extension in the full phase space $\mathbb{R}^{2N}$, rather than analyticity restricted to the physical domain.}

\section{Exhaustive set of analytic conserved quantities for free motion}\label{sec:freemotionU1}
We first construct the most general analytic conserved quantity $Q \in \mathbb{C}[\{x_i\}, \{p_i\}]$ of free motion in a harmonic trap {consisting of non-interacting rods that do not collide but pass through each other}.
It is easy to verify that the lowest-degree conserved quantities are the following two families:
\begin{equation}
    C_{ij}^+ = x_i x_j + p_i p_j,\;\;\;\; C_{ij}^- = x_i p_j - p_i x_j.
\label{eq:Q2families}
\end{equation}
These generate the algebra of conserved quantities $\mathbb{C}[\{C_{ij}^+, C_{ij}^-\}]$.
We then wish to ask if this is the \textit{exhaustive} set of conserved quantities for free motion. 
Below, we show that this is indeed the case using the well-known idea that free harmonic motion corresponds to a rotation in each $x{_i}-p{_i}$ plane.
Thus the problem reduces to that of finding functions invariant under such rotations in all the $x{_i}-p{_i}$ planes, or the invariants of the $SO(2) \cong U(1)$ group.
It is already a known result in the representation theory of groups~\cite{weyl1997} that any such invariant can be constructed out of the scalar product $x_i x_j+p_i p_j$ and the cross product/determinant $x_ip_j-x_jp_i$.
We nevertheless prove it again here for completeness.
\begin{theorem}
    \label{th:freemotion}
    Any $Q \in {\mathbb{C}[\{x_i\}, \{p_i\}]}$ that is conserved during free motion according to Eq.~(\ref{eq:freeconserved}) {is $U(1)$ invariant under the operation $(x_i, p_i) \rightarrow (x_i \cos\theta + p_i\sin\theta, -x_i\sin\theta + p_i\cos\theta)$, and} is generated by the degree 2 conserved polynomials of Eq.~(\ref{eq:Q2families}), i.e., $Q \in \mathbb{C}[\{C_{ij}^+, C_{ij}^-\}]$.
    Equivalently, in the complex coordinates, any $Q \in \mathbb{C}[\{z_i\}, \{\bar z_i\}]$ that is conserved during free motion satisfies $Q \in \mathbb{C}[\{z_i \bar z_j\}]$, which is $U(1)$ invariant as $(z_i, \bar z_i) \rightarrow (e^{-\mathrm{i}\theta}z_i, e^{\mathrm{i}\theta}\bar z_i)$.
\end{theorem}
\begin{proof}
Working with complex coordinates of Eq.~(\ref{eq:z_def}), any general $Q \in \mathbb{C}[\{x_i\}, \{p_i\}] = \mathbb{C}[\{z_i\}, \{\bar z_i\}]$ can be written as
\begin{equation}\label{eq:general_z_barz}
Q(\{z_i\},\{\bar z_i\})= \sum_{{\bf r}, {\bf s}}
C_{\mathbf r,\mathbf s} M_{\mathbf r, \mathbf s}(\{z_i\}, \{\bar z_i\}),\;\;\;\;M_{\mathbf{r}, \mathbf{s}}(\{z_i\}, \{\bar z_i\}) \defn \prod_{i=1}^N z_i^{r_i}\bar z_i^{s_i},
\end{equation}
where ${\bf r}$ and ${\bf s}$ are $N$-dimensional vectors with entries $\{r_i\}$ and $\{s_i\}$ in $\mathbb{Z}_{\geq 0}$, and the coefficients $C_{\mathbf{r}, \mathbf{s}}$ are some complex numbers.
Eq.~(\ref{eq:freemotion}) in complex coordinates reads
\begin{equation}\label{eq:z_eom}
\dot z_i=-\mathrm{i}z_i,\qquad \dot{\bar z}_i=\mathrm{i}\bar z_i,\;\;\;\implies\;\;\;{z_i(t) = e^{-\mathrm{i}t} z_i(0),\;\;\;\bar z_i(t) = e^{\mathrm{i}t} \bar z_i(0).} 
\end{equation}
{Note that in the phase space coordinates this is equivalent to
\begin{equation}
    x_i(t) = x_i(0) \cos t + p_i(0) \sin t,\;\;\;p_i(t) = -x_i(0) \sin t + p_i(0) \cos t.
\label{eq:U1phasespace}
\end{equation}}
Taking the time-derivative of \eqref{eq:general_z_barz} 
and using \eqref{eq:z_eom} gives
\begin{equation}\label{eq:L_def}
\frac{dQ}{dt}
=-\mathrm{i}\sum_i\Big(z_i\frac{\partial Q}{\partial z_i}-\bar z_i\frac{\partial Q}{\partial\bar z_i}\Big) \equiv -\mathrm{i}\,\mathcal{L}[Q],\;\;\;\;\mathcal{L}\defn\sum_{i=1}^N\Big(z_i\frac{\partial}{\partial z_i}-\bar z_i\frac{\partial}{\partial\bar z_i}\Big).
\end{equation}
where we have defined the linear first-order differential operator $\mathcal{L}$.
Conservation of $Q$ (i.e., $\dot Q=0$) is therefore equivalent to the condition $\mathcal{L}[Q]=0$.
The monomials $M_{\mathbf r, \mathbf s}$ in Eq.~(\ref{eq:general_z_barz})
are eigenfunctions of $\mathcal{L}$, and hence we have
\begin{equation}\label{eq:L_on_monomial}
\mathcal{L}\,M_{\mathbf r,\mathbf s}=(|\mathbf r|-|\mathbf s|)\,M_{\mathbf r,\mathbf s} = 0,\;\;\;\implies\;\;|\mathbf r| = |\mathbf s|,\;\;\;\;\;\;|\mathbf r|:=\sum_{i=1}^N r_i,\;|\mathbf s|:=\sum_{i=1}^N s_i.
\end{equation}
Since these monomials are all linearly independent, any $Q$ that is annihilated by $\mathcal{L}$ is necessarily the sum of \textit{balanced} monomials, i.e., those with the same $z$- and $\bar z$-degree:
\begin{equation}\label{eq:conserved_general}
{Q(\{z_i\}, \{\bar z_i\})}=\sum_{|\mathbf r|=|\mathbf s|}
c_{\mathbf r,\mathbf s}\;\prod_{i=1}^N z_i^{r_i}\bar z_i^{s_i}\;\;\in \mathbb{C}[\{z_i \bar z_j\}],
\end{equation}
where in the second step we have used an elementary and useful observation that the $N^2$ monomials $\{z_i\bar z_j\}$ generate all balanced monomials.
This reduction can also be viewed as a direct consequence of the $U(1)$ phase-rotation symmetry in Eq.~(\ref{eq:z_eom}).
$-\mathrm{i}\mathcal{L}$ (with $\mathcal{L}$ defined in Eq.~\ref{eq:L_def}) is just the infinitesimal generator of the $SO(2)\cong U(1)$ acting on functions, i.e.,$e^{-\mathrm{i}\theta\mathcal{L}}f(z,\bar{z})=f(e^{-\mathrm{i}\theta}z,e^{\mathrm{i}\theta}\bar{z})$, and functions satisfying $\mathcal{L}f=0$ are simply those invariant under the $SO(2)$ or $U(1)$ transformation of Eqs.~(\ref{eq:z_eom}) {and (\ref{eq:U1phasespace})}.
Noting that 
\begin{equation}
    z_i\bar{z}_j=x_ix_j+p_ip_j+\mathrm{i}(p_ix_j-x_ip_j)\;\;\;\implies\;\;\; \mathbb{C}[\{z_i \bar z_j\}] = \mathbb{C}[\{C_{ij}^+, C_{ij}^-\}]
\label{eq:balancedpoly}
\end{equation}
completes the proof that every $Q \in \mathbb{C}[\{z_i\}, \{\bar z_i\}]$ conserved under free motion is generated by the elementary polynomials of Eq.~(\ref{eq:Q2families}).
\end{proof}
{This system of free particles is also superintegrable, which means that it has $2N-1$ functionally independent integrals of motion, which is the maximum allowed for any dynamical system of $2N$ variables.}
$N$ of them are {the} energies of each particle $E_i= C_{ii}^+ = x_i^2+p_i^2$, but there are more such as $C_{1j}^+ =x_1 x_j+p_1 p_j$ for $2 \leq j \leq N${, which are shown in Eq.~(\ref{eq:Q2families})}.
These form a total of $N+N-1=2N-1$ independent integrals of motion, since in polar phase space coordinates {these read}
\begin{equation}
x_i=r_i\cos\theta_i,\;\; p_i=r_i\sin\theta_i,\;\;\;z_i = r_i e^{i\theta_i}\;\;\implies\;\;E_i=r_i^2,\;\;C^+_{1j}=r_1r_j\cos(\theta_j-\theta_1), 
\label{eq:polarcoordinates}
\end{equation}
{and their} Jacobian {of Eq.~(\ref{eq:jacobian})} with respect to $\{r_i,\theta_i\}$ has full rank {of $2N-1$ for generic values of $\{r_j\}$ and $\{\theta_j\}$}.
This makes it maximally superintegrable, as also follows as a special case of~\cite{RodriguezTempestaWinternitz2008}}. 

\section{{Exhaustive set of analytic conserved quantities for zero length rods (point particles) with collisions}}
\label{sec:a=0}
We now discuss the case when all the rods are point particles, i.e., $a_n = 0$ for all $n$.
In this case, since the collisions between the particles are elastic, we know that the system should be almost identical to the case of free motion of particles where the different particle pass through each other. 
The only difference is that the particles here get reflected off each other, which requires keeping track of the labels of each of the particles.
As we discuss below, this enforces an $S^{(x,p)}_N$ permutation symmetry on the conserved quantities, where this is the permutation symmetry of the labels, which corresponds to joint permutations of the positions and the momenta, i.e., $(x_i, p_i) \rightarrow (x_{\sigma(i)}, p_{\sigma(i)})$ {for $\sigma \in S_N$}.
%
%
Knowing that any conserved quantity $Q$ should necessarily be conserved under free motion, we start by showing the following Lemma that applies to collisions.
\begin{lemma}
\label{lem:a=0}
    When $a_n = a_{n+1} = 0$, any $Q\in \mathbb{C}[\{C_{ij}^+, C_{ij}^-\}] = \mathbb{C}[\{z_i \bar z_j\}]$ that satisfies Eq.~(\ref{eq:collisioncondition}) {should necessarily be symmetric under the exchange of the labels $n$ and $n+1$, i.e., the simultaneous exchange of the positions and momenta of the rods $n$ and $n+1$}.
\end{lemma}
\begin{proof}
Since $Q$ is conserved during free motion, we can first expand it in terms of homogeneous polynomials as $Q=\sum_qQ_{2q}$, where $Q_{2q}$ is a balanced polynomial with total degree $2q$, which also means that it has \textit{bidegree} $(q,q)$ [which are the degrees under the $\{z_i\}$ and $\{\bar z_i\}$]. 
Imposing Eq.~(\ref{eq:collisioncondition}) for the collision of rods $n$ and $n+1$, {and noting that $b_n = 0$}, we obtain
\begin{equation}
    \sum_{q}\mathcal{X}_n^{(0)}[Q_{2q}] = \sum_{q}\mathcal{P}_n\mathcal{X}_n^{(0)}[Q_{2q}]\;\;\;\implies\;\;\;\mathcal{X}_n^{(0)}[Q_{2q}]=\mathcal{P}_n\mathcal{X}_n^{(0)}[Q_{2q}],
\label{eq:a=0condition}
\end{equation}
where we have used the fact that all the terms in the sums have different degrees and thus are linearly independent.
We now express Eq.~(\ref{eq:a=0condition}) in terms of the complex variables $\{z_n\}$ and $\{\bar{z}_n\}$.
First, we note that the action of $\mX^{(0)}$ is just setting $x_{n + 1} = x_n$, and can be expressed as
\begin{equation}
    \Re{z_{n+1}} = \Re{z_n} \;\;\implies\;\;L_n \defn (z_n-z_{n+1})+(\bar{z}_n-\bar{z}_{n+1})=0.
\label{eq:realcondition}
\end{equation}
Second, since $x_{n+1} = x_{n}$ when $L_n = 0$, Eq.~(\ref{eq:a=0condition}) can equivalently be written as
\begin{equation}
    g_n \defn Q_{2q}-\sigma_n [Q_{2q}] = 0 \;\;\;\text{if}\;\;L_n = 0,
\label{eq:gnpolynomial}
\end{equation}
where we have defined $\sigma_n$ as the operator that exchanges the rod indices $n$ and $n+1$, i.e., exchanges both the momenta \textit{and} the positions of the rods.
By the factor theorem for polynomials, Eq.~(\ref{eq:gnpolynomial}) implies 
\begin{equation}
    g_n=L_n\cdot U(\{z_i\},\{\bar{z}_i\}),\;\;\;U \in {\mathbb{C}[\{z_i\},\{\bar z_i\}]}.
\label{eq:gnform}
\end{equation}
Since both $Q_{2q}$ and $\sigma_n[Q_{2q}]$ are polynomials of bi-degree $(q,q)$, $g_n$ should also be a homogeneous polynomial of bi-degree $(q,q)$ unless it is zero.
In the following paragraph, we show that this is just the zero polynomial.
It is convenient to define the $\delta$-degree of monomials as 
\begin{equation}
    \delta\left(\prod\nolimits_{j=1}^N{z_j^{r_j}}\prod\nolimits_{j=1}^N{\bar z_j^{s_j}}\right) \defn \sum_j{(r_j - s_j)}.
\label{eq:deltadegree}
\end{equation} 
{ A polynomial is called $\delta$-homogeneous of degree $t$ if all of its nonzero monomials have $\delta$-degree $t$. For a nonzero $\delta$-homogeneous polynomial $F$, we therefore write $\delta(F)=t$. With this definition, $Q_{2q}$ and $g_n$ are $\delta$-homogeneous with $\delta(Q_{2q})=\delta(g_n)=0$.}
We can then write
\begin{equation}
L_n \defn L^+ + L^-,\;\;\;L^+ \defn z_n-z_{n+1},\;\;\;L^-\defn \bar z_n -\bar z_{n+1},\;\;\;\delta(L^{\pm}) = \pm 1, 
\end{equation}
and decompose \(U\) of Eq.~(\ref{eq:gnform}) into its \(\delta\)-homogeneous components {to obtain}:
\begin{equation}
U=\sum_{t = T_{\rm min}}^{T_{\rm max}} U_t,\;\;\;\delta(U_t) = t\;\;\;\implies\;\;\; g_n=L_n U=\sum_{t = T_{\rm min}}^{T_{\rm max}}\big(L^+U_t+L^-U_t\big),\;\;\delta(L^+U_t)=t+1,\ \delta(L^-U_t)=t-1,
\label{eq:Udecomposition}
\end{equation}
{the existence of $T_{\rm min} \leq T_{\rm max}$ is guaranteed by the fact that $g_n$ has a finite degree $2q$.}
{Now note that} the components of \(g_n\) of \(\delta\)-degree \(T_{\rm min} - 1\) and \(T_{\rm max} + 1\) are exactly \(L^-U_{T_{\rm min}}\) and \(L^+U_{T_{\rm max}}\).
However, since \(\delta(g_n)=0\), this should mean $T_{\rm min} = 1$ and $T_{\rm max} = -1$, which is a contradiction, implying that $U = 0$.
This also shows that $g_n = 0$. 

Hence, conservation under Eqs.~(\ref{eq:freeconserved}) and (\ref{eq:collisionconserved}) implies that
\begin{equation}
    Q_{2q}=\sigma_n[Q_{2q}]\;\;\;\implies\;\;\;Q = \sigma_n[Q],
\label{eq:condition1}
\end{equation}
which completes the proof that $Q$ is symmetric under label exchange. 
\end{proof}

{When $a_n = 0$ for all $n$, this Lemma implies that the conserved quantities are invariant under the permutation group $S^{(x,p)}_N$ of all possible permutations of the $N$ rod indices.
In this case, we obtain the following theorem.}
\begin{theorem}\label{thm:a=0}
When $a_n = 0$ for all $n$, any $Q \in \mathbb{C}[\{z_i \bar z_j\}]$ that satisfies Eq.~(\ref{eq:collisionconserved}) for all $n$ also satisfies
    \begin{equation}
        Q \in \mathbb{C}\left[\left\{\prod_{\alpha=1}^R Z_{m_{\alpha},n_{\alpha}},\;\;\;\sum_{\alpha=1}^Rm_{\alpha}=\sum_{\alpha=1}^Rn_{\alpha}\right\}\right],\;\;\;Z_{m,n}\defn\sum_{j=1}^Nz_j^m\bar{z}_j^n.
    \label{eq:a=0ring}
    \end{equation}
\end{theorem}
\begin{proof}
Applying Lemma \ref{lem:a=0} for all $n$ shows that $Q$ is invariant under the exchange of any two adjacent indices, which means that under the $S^{(x,p)}_N$ group generated by joint permutations of the positions and momenta $\{x_i,p_i\}\to\{x_{\sigma(i)},p_{\sigma(i)}\}$ for any $\sigma\in S_N$.
We now make the set of conserved quantities explicit using their invariance under this $S^{(x,p)}_N$ symmetry and the $U(1)$ symmetry imposed by free motion.
A convenient way to write all such symmetric conserved quantities is by symmetrizing balanced monomials.
Given any balanced monomial $M_{\vect{r},\vect{s}}(\{z_i\},\{\bar z_i\})$ of the form of Eq.~(\ref{eq:general_z_barz}), we can define its symmetrization as
\begin{equation}
\mathrm{Sym}\, M_{\vect{r},\vect{s}}
:=
\sum_{\sigma\in S_N}
\prod_{j=1}^N z_{\sigma(j)}^{\,r_j}\bar z_{\sigma(j)}^{\,s_j}.
\label{eq:sym_balanced_monomial}
\end{equation}
Each such symmetrization is manifestly $S^{(x,p)}_N$ and $U(1)$ invariant, and hence is a conserved quantity for the point particle dynamics. 
Conversely, every homogeneous polynomial conserved for this point particle dynamics can be written as a linear combination of such symmetrized balanced monomials.
An equivalent and often more compact description is in terms of the $S_N$ symmetric power sums
\begin{equation}
Z_{m,n}:=\sum_{j=1}^N z_j^{\,m}\bar z_j^{\,n}.
\label{eq:pmn_def}
\end{equation}
These are not balanced, hence \(Z_{m,n}\) is not conserved by itself unless \(m=n\), but the algebra generated by their balanced product, i.e., 
\begin{equation}
\left\{\prod_{\alpha=1}^R Z_{m_\alpha,n_\alpha}\right\},\;\;\;\sum_{\alpha=1}^R m_\alpha=\sum_{\alpha=1}^R n_\alpha,
\label{eq:product_pmn}
\end{equation}
is always conserved. 
This completes the proof that $Q$ belongs to the algebra shown in Eq.~(\ref{eq:a=0ring}). 
\end{proof}
This makes explicit that the \(a=0\) case has a much larger family of conserved quantities than just $E$ and $E_{\rm cm}$. 
A few low-degree examples that are conserved in this case are:
\begin{align}
Z_{1,1} &= E = \sum_i |z_i|^2
= \sum_i (x_i^2+p_i^2),\;\;Z_{1,0} Z_{0,1} = E_{\rm cm} = |\sum_i z_i|^2
= (\sum_i x_i)^2+(\sum_i p_i)^2, \nonumber \\
Z_{2,2} &= \sum_i |z_i|^4,\;\;\;Z_{2,0}Z_{0,2} = |\sum_i z_i^2|^2,\;\;\;
Z_{2,1}Z_{0,1} = (\sum_i z_i^2\bar z_i)(\sum_j \bar z_j),
\end{align}
The first line just shows $E$ and $E_{\rm cm}$, whereas the second line shows some quantities that can be proved to be independent of them {using the Jacobian of Eq.~(\ref{eq:jacobian})}.
{Due to its similarity with the system of non-interacting particles} in a harmonic trap, {we expect that this system} is {maximally} superintegrable {as well}.
{Indeed, we can explicitly} construct the $2N-1$ independent conserved quantities for this system.
Recall the
symmetric power sums of Eq.~(\ref{eq:pmn_def}) that are conserved in this system, and construct the following $2N-1$ real polynomials:
\begin{align}
    I_k \defn Z_{k,k},\;\;1 \leq k \leq N,\;\;\;J_k \defn \operatorname{Im}\!\left(
        Z_{k+1,k}Z_{0,1}
    \right),\;\;\;1 \leq k \leq N-1.
    \label{eq:point-particle-IkJk}
\end{align}
To establish their functional independence, we {again use the} polar coordinates of Eq.~(\ref{eq:polarcoordinates}), where we have
\begin{equation}
    I_k=\sum_{j=1}^{N}r_j^{2k},\;\;\;Z_{k+1,k}
    =
    \sum_{j=1}^{N}
    r_j^{2k+1}e^{i\theta_j},
    \;\;\;
    Z_{0,1}
    =
    \sum_{j=1}^{N}
    r_j e^{-i\theta_j}\;\;\implies\;\;J_k = \sum_{i,j = 1}^N{r_i^{2k+1} r_j \sin(\theta_i - \theta_j).}
\end{equation}
Hence the Jacobian {in polar coordinates has a block lower triangular structure and the following inequality:}
\begin{equation}
    \mathcal{J} = \begin{pmatrix}
        V & 0\\
        U & W
    \end{pmatrix}\;\;\;\implies\;\;\text{rank}(\mathcal{J}) \geq \text{rank}(V) + \rank(W).
\label{eq:jacobianstructure}
\end{equation}
{Here $V$ is an $N \times N$ matrix corresponding to $\{I_k\}$ and $\{r_j\}$, and $W$ is an $(N-1) \times N$ matrix corresponding to $\{J_k\}$ and $\{\theta_j\}$, and $U$ is an $(N-1) \times N$ matrix corresponding to $\{J_k\}$ and $\{r_j\}$.}
$V$ {can be directly be computed to be} a Vandermonde-type matrix with {with elements $V_{kj}= \frac{\partial I_k}{\partial r_j} = 2k r_j^{2k-1}$, whose} determinant {is generically non-zero, which implies that $\text{rank}(V) = N$.}
For $W$, we {can compute its} elements {to be}
\begin{equation}
    W_{kj} = \frac{\partial J_k}{\partial\theta_j}
    =
    r_j
    \left(
        A_{0j} r_j^{2k}-A_{kj}
    \right),\;\;\;A_{kj}=\sum_{\ell=1}^{N}{r_\ell^{2k+1}{\cos(\theta_\ell - \theta_j)}},
    \qquad 0 \leq k \leq N-1.
    \label{eq:J-angle-derivative}
\end{equation}
We now show that $\text{rank}(W) = N - 1$ by showing that $\rank(W_0) = N - 1$, where $W_0$ is $W$ evaluated at $\theta_1 = \cdots \theta_N = 0$.
At this point $A_{kj}$ becomes independent of $\{\theta_i\}$.
If $\operatorname{rank}(W_0) < N-1$, {there should exist a} linear combination of its rows {that} vanishes, i.e., for every $j$, we should have
\begin{equation}
    \sum_{k=1}^{N-1}
    c_k(\{r_l\})
    \left(A_{0j} r_j^{2k}-A_{kj}\right)=0,
\label{eq:lincomb}
\end{equation}
where $c_k$ can be a function of $\{r_\ell\}$.
We can now define a polynomial in a continuous variable $X$:
\begin{equation}
P(X) = A_0 \sum_{k=1}^{N-1} c_k(\{r_l\}) X^k - \sum_{k=1}^{N-1} c_k(\{r_l\}) A_k, 
\end{equation}
which should satisfy $P(r_j^2) = 0$ for all $N$ values of $j$, {where $\{r_j^2\}$ can all take $N$ distinct positive values}.
However, {since} the polynomial $P(X)$ has a degree of $N-1$, {it cannot have $N$ distinct roots unless} it is identically the zero polynomial ($P(X) \equiv 0$), {with which we can conclude $c_k = 0$ for all $k$, i.e., no non-trivial linear combination that satisfies Eq.~(\ref{eq:lincomb}) exists.}
{Hence} the rows of $W_0$ are linearly independent, and we {can} conclude {that}:
\begin{equation}
\operatorname{rank}(W_{0}) = N-1 {= \text{rank}(W)},
\end{equation}
{where the last equality follows because $N-1$ is the maximum possible rank of $W$.
With this and Eq.~(\ref{eq:jacobianstructure}), we get that $\text{rank}(\mathcal{J}) = 2N - 1$, which is its maximum possible rank. 
}
$\{I_k\}$ and $\{J_k\}$ are therefore $(2N-1)$ functionally independent conserved quantities, making the system maximally superintegrable.
\section{Absence of additional analytic conserved quantities for Stochastic Momentum Exchange Dynamics (SMED)}\label{sec:SMED}
We now consider a slightly simpler process considered in~\cite{bagchi}, termed the Stochastic Momentum Exchange Dynamics (SMED), where the adjacent rods randomly exchange momenta at any point in time in addition to following their deterministic dynamics (i.e., free motion and collisions).
Any conserved quantity $Q$ under such a process should be conserved under free motion, i.e., satisfy Eq.~(\ref{eq:freeconserved}) and should be invariant under exchanges of momenta of adjacent pairs {for arbitrary positions (not just at the time of collision)}, and thus under \textit{arbitrary} permutations of the $N$-momenta $\{p_1,p_2,...,p_N\}$, i.e., 
\begin{equation}
Q(\{x_1,\dots,x_N\},\{p_1,\dots,p_N\})
=
Q(\{x_1,\dots,x_N\}, \{p_{\sigma(1)},\dots,p_{\sigma(N)}\})
\qquad \forall \sigma\in S_N.
\label{eq:SMEDconserved}
\end{equation}
We say that such functions are invariant under the group $S^{(p)}_N$.
Note that the collision condition of Eq.~(\ref{eq:collisionconserved}) is automatically satisfied by any $Q$ that satisfies Eq.~(\ref{eq:SMEDconserved}), hence we need not consider it separately, and the rod lengths just do not enter the analysis here.
We now show that given $Q$ is conserved under free motion, i.e., $Q \in \mathbb{C}[\{z_i \bar z_j\}]$ the only independent conserved quantities in the SMED are $E$ and $E_{\rm cm}$ of Eq.~(\ref{eq:traditionalconserved}) or (\ref{eq:EEcmdefn}).

\begin{theorem}\label{thm:Pexchange}
{Any $Q \in \mathbb{C}[\{C_{ij}^+, C_{ij}^-\}] = \mathbb{C}[\{z_i \bar z_j\}]$ that is symmetric under arbitrary exchange of momenta, i.e., satisfies Eq.~(\ref{eq:SMEDconserved}), is also invariant under arbitrary exchange of positions.
Furthermore, it is functionally dependent on $E$ and $E_{\rm cm}$, i.e., $Q \in \mathbb{C}[E, E_{\rm cm}]$.}
\end{theorem}
\begin{proof}
Given that $Q \in \mathbb{C}[\{z_i \bar z_j\}]$, it is $U(1)$ symmetric under the transformations {[see Eqs.~(\ref{eq:z_eom}) and (\ref{eq:U1phasespace})]}
\begin{equation}
    (z_j, \bar z_j) \mapsto (e^{-it} z_j, e^{it} \bar z_j)\;\;\;\forall\;\;j\;\;\;\implies\;\;\;(x_j, p_j) \mapsto (x_j \cos t + p_j \sin t, -x_j\sin t + p_j\cos t)\;\;\;\forall\;\;j,  
\label{eq:U1transform}
\end{equation}
where $t$ can also be interpreted as ``time'' under the free-motion equations of motion, but we will not use that interpretation here.
First, we notice that we can use \(t=\pi/2\) in Eq.~(\ref{eq:U1transform}) to show that any $U(1)$ invariant $Q$ that is $S^{(p)}_N$ invariant should also be $S^{(x)}_N$ invariant under the permutations of the positions $\{x_j\}$, since that gives $(x_i,p_i)\mapsto (p_i,-x_i)$, which should mean that 
\begin{equation}
Q(\{x_{\sigma(i)}\},\{p_i\}) = Q(\{p_i\},\{-x_{\sigma(i)}\})= Q(\{p_{i}\},\{-x_i\}) = Q(\{x_i\},\{p_i\})\;\;\;\forall\;\;\sigma \in S_N. 
\label{eq:SpSxinvariance}
\end{equation}
Thus $Q$ is symmetric in the $\{x_i\}$ and $\{p_i\}$ separately, or equivalently it is invariant under the group $S_N^{(x)} \times S_N^{(p)}$, and a $U(1)$ group that interacts non-trivially with the permutations.
We consider the action of this symmetry group on the $2N$-dimensional space spanned by $\{x_i\}$ and $\{p_i\}$, which are the generators of the complete algebra of quantities under consideration.
As we show, understanding how this space of generators transform under the symmetry will allow us to construct conserved quantities invariant under this symmetry. 
In particular, in this space, we show that the action of the symmetry splits irreducibly into a $2$-dimensional representation of an $SO(2)$ and an $(2N-2)$-dimensional representation of an $SO(2N-2)$ group, whose invariants are well-known. 
To do so, we start with the $2N$-dimensional matrix representation of the infinitesimal generator of rotations of Eq.~(\ref{eq:U1transform}), which is given by
\begin{equation}
    J_I = \begin{pmatrix}
        0 & -I \\
        I & 0
    \end{pmatrix},
\label{eq:U1gen}
\end{equation}
where we have ordered the basis elements as $\{x_1, \cdots, x_N, p_1, \cdots, p_N\}$ and $I$ is the $N$-dimensional identity matrix.
In the same basis ordering, the elements of the discrete permutation group $S_N^{(x)} \times S_N^{(p)}$ are the block matrices 
\begin{equation}
    \Pi_{\Gamma, \Delta} = \begin{pmatrix} \Gamma & 0 \\ 0 & \Delta \end{pmatrix},
\label{eq:SNgen}
\end{equation}
where $\Gamma$ and $\Delta$ are independent $N$-dimensional permutation matrices (hence also orthogonal matrices).
This means that additional infinitesimal generators can be obtained by conjugating $J_I$ with each $\Pi_{\Gamma,\Delta}$, i.e., 
\begin{equation}
    J_{\Sigma} \defn \begin{pmatrix} 0 & -\Sigma \\ \Sigma^T & 0 \end{pmatrix} = \Pi_{\Gamma,\Delta} J_I \Pi_{\Gamma,\Delta}^{T},\;\;\;\;\Sigma \defn \Gamma \Delta^T,
\label{eq:JSigmadefn}
\end{equation}
where $\Sigma$ is an $N$-dimensional permutation matrix, just as $\Gamma$ and $\Delta$.
Hence the full Lie group is generated by all the infinitesimal generators $\{J_{\Sigma}\}$ for all $N!$ permutation matrices $\Sigma$.
We now characterize the real Lie algebra generated by these matrices.
We first notice that the $N$-dimensional uniform vector $\mathbf{1} = (1, \dots, 1)^T$ is an eigenvector of all $N$-dimensional permutation matrices $\Sigma$.
Hence, we can decompose any permutation matrix as in this one-dimensional space and its $(N-1)$-dimensional orthogonal complement as
\begin{equation}
    \Sigma = 1 \oplus M_{\Sigma} \;\;\;\text{where}\;\;\;M_\Sigma \in M_{N-1}(\mathbb{R}), 
\label{eq:Sigmadecomp}
\end{equation}
where $M_{k}(\mathbb{R})$ is the space of all $k$-dimensional real matrices.
Further, using the Birkhoff-von Neumann theorem, we can show that the linear span of permutation matrices are doubly stochastic matrices up to an overall rescaling and additive factor of $\mathbf{1}\mathbf{1}^T$ (see, e.g., proposition 3.1 of~\cite{Anderson2021}), i.e., it has the form:
\begin{equation}
    \mathcal{S} \defn \text{span}_{\mathbb{R}}\{\Sigma\} = \{A \in M_N(\mathbb{R}):\;\;A\mathbf{1} = \lambda\mathbf{1},\;\;\mathbf{1}^T A = \lambda\mathbf{1}^T\}, 
\label{eq:permspan}
\end{equation}
and we can directly count its dimension to be $\text{dim}(\mathcal{S}) = 1 + (N-1)^2$. 
This is exactly the dimension of block-diagonal matrices of the form in Eq.~(\ref{eq:Sigmadecomp}), so we can also conclude that 
\begin{equation}
    \text{span}_{\mathbb{R}}\{M_\Sigma\} = M_{N-1}(\mathbb{R}). 
\label{eq:MSigmaspan}
\end{equation}
This leads to a decomposition of $J_\Sigma$ into {two parts: (i)} The 2-dimensional component spanned by [in the basis of Eq.~(\ref{eq:JSigmadefn})] the $2N$-dimensional vectors $(\mathbf{1}, \mathbf{0})^T$ and $(\mathbf{0}, \mathbf{1})^T$ where $\mathbf{0}$ is the $N$-dimensional zero vector, and {(ii)} its $(2N-2)$-dimensional component is its orthogonal complement, as
\begin{equation}
     J_\Sigma = \begin{pmatrix}
    0 & -1 \\
    1 & 0 
    \end{pmatrix} \oplus \begin{pmatrix}
    0 & -M_{\Sigma} \\
    M_{\Sigma}^T & 0 
    \end{pmatrix}. 
\label{eq:JSigmadecomp}
\end{equation}
Taking commutators of $\{J_\Sigma\}$ among themselves and their linear span, and using Eq.~(\ref{eq:MSigmaspan}) shows that the real Lie algebra generated by $\{J_\Sigma\}$ is simply 
\begin{equation}
    \text{Lie}_{\mathbb{R}}(\{J_\Sigma\}) = \mathfrak{so}(2) \oplus \mathfrak{so}(2N-2),
\end{equation}
where the $\mathfrak{so}(k)$ is the Lie algebra of all $k$-dimensional skew-symmetric matrices.
Hence the full symmetry group contains the group $SO(2) \times SO(2N-2)$, and we now use this information to construct invariant functions. 

In the language of $\{x_i\}$ and $\{p_i\}$, the two-dimensional subspace invariant under the full symmetry group is spanned by the normalized center-of-mass coordinates $\{x_{\rm cm}/\sqrt{N}, p_{\rm cm}/\sqrt{N}\}$, where $x_{\rm cm} := \sum_{j=1}^N x_j$ and $p_{\rm cm} := \sum_{j=1}^N p_j$. Its $(2N-2)$-dimensional orthogonal complement can always be spanned by choosing $N-1$ "relative" coordinates $\{x'_1, \dots, x'_{N-1}\}$ and their conjugate momenta $\{p'_1, \dots, p'_{N-1}\}$ such that the full set of $2N$ variables forms an orthonormal basis. 
The symmetry group acts as $SO(2) \cong U(1)$ on the center-of-mass coordinates and as $SO(2N - 2)$ on the relative coordinates.
Any vector space of $n$ independent variables acted upon transitively by $SO(n)$ possesses a single independent quadratic invariant corresponding to its quadratic norm~\cite{weyl1997}.
Applying this fundamental theorem of invariant theory to the two decoupled sectors, we find that the only independent invariants are $I_0$ and $I_1$, defined as:
\begin{equation}
I_0 = \frac{1}{N}\left(x_{\rm cm}^2 + p_{\rm cm}^2\right), \quad I_1 := \sum_{i=1}^{N-1} \left[(x'_i)^2 + (p'_i)^2\right].
\end{equation}
It is clear that these are related to $E$ and $E_{\rm cm}$
\begin{equation}
    I_0 = \frac{E_{\rm cm}}{N},\;\;\;I_0 + I_1 = \sum_{j=1}^N{(x_j^2 + p_j^2)} = E,
\end{equation}
where in the second equation we have used the fact that the basis change {from $\{x_i\}$ and $\{p_i\}$ } to center-of-mass and relative coordinates {$\{x_{\rm cm}, p_{\rm cm}, \{x_i'\}, \{p_i'\}\}$} is an orthogonal transformation. 
Hence the full algebra of conserved quantities can also just be expressed as $\mathbb{C}[E, E_{\rm cm}]$.
Further, it is easy to see that $E$ and $E_{\rm cm}$ are evidently invariant under $S^{(x)}_N$ and $S^{(p)}_N$, hence this is the complete set of conserved quantities, which completes the proof. 
\end{proof}
\section{Absence of additional analytic conserved quantities with collisions between rods of non-zero length}
\label{sec:absenceofthird}
We now apply the results of Secs.~\ref{sec:freemotionU1}, \ref{sec:a=0}, and \ref{sec:SMED} to show that when {a single rod has a non-zero length, i.e., $a_n \neq 0$ for some $n$,} any conserved $Q$ during free motion and collisions is functionally dependent on $E$ and $E_{\rm cm}$.
We start with the fact that any $Q$ that is conserved during free motion has to be balanced (i.e., in the algebra $\mathbb{C}[\{z_i \bar z_j\}]$), and first show the following Lemma. 
\begin{lemma}
     \label{lem:collisions}
When { $b_n$ defined in Eq.~(\ref{eq:collision})} is non-zero, {} {any} $U(1)$ invariant $Q\in \mathbb{C}[\{C_{ij}^+, C_{ij}^-\}] = \mathbb{C}[\{z_i \bar z_j\}]$ that satisfies Eq.~(\ref{eq:collisioncondition}) should be invariant under the exchange of the momenta of the rods $n$ and $n+1$.
\end{lemma}
\begin{proof}
Since $Q \in \mathbb{C}[\{z_i \bar z_j\}]$ by virtue of conservation under free motion, $Q$ should be balanced in the complex variables $\{z_n\}$ and $\{\bar{z}_n\}$. 
We impose this balancing condition to Eq.~(\ref{eq:collisioncondition}) by first defining complex variables as follows
\begin{equation}
    u_n \defn z_{n+1} - z_n,\;\;\;\bar u_n \defn \bar z_{n+1} - \bar z_n,\;\;\;v_n \defn z_{n+1} + z_n,\;\;\;\bar v_n \defn \bar z_{n+1} + \bar z_n,
\end{equation}
which satisfy the conditions
\begin{equation}
    \mathcal{P}_n[u_n] = \bar{u}_n,\;\;\;\mathcal{P}_n[\bar u_n] = u_n,\;\;\;\mathcal{P}_n[v_n] = v_n,\;\;\;\mathcal{P}_n[\bar v_n] = \bar{v}_n,
\label{eq:Pnaction}
\end{equation}
where $\mathcal{P}_n$ is the momentum exchange operator defined in Eq.~(\ref{eq:superopdefn}).
Working in terms of the variables $\{z_{j \neq n, n+1}\}$, $\{\bar{z}_{j \neq n, n+1}\}$, $u_n$, $\bar{u}_n$, $v_n$, $\bar{v}_n$, the collision condition of Eq.~(\ref{eq:collisioncondition}) reads
\begin{equation}
    Q(\{\cdots, u_n, \bar u_n, \cdots \}) = Q(\{\cdots, \bar u_n, u_n, \cdots \})\;\;\;\text{if}\;\;\;u_n + \bar u_n = 2b_n\;\;\;\implies\;\;\;Q(\{\cdots, u_n, 2b_n - u_n, \cdots \}) = Q(\{\cdots, 2b_n - u_n, u_n, \cdots \}),
\label{eq:Qcondition}
\end{equation}
where $\cdots$ denotes the rest of variables apart from $u_n$ and $\bar u_n$. 
Using the $\delta$-degree of Eq.~(\ref{eq:deltadegree}), we have that $\delta(u_n) = \delta(v_n) = 1$ and $\delta(\bar{u}_n) = \delta(\bar{v}_n) = -1$, and we can always separate the dependence of $Q$ on $u_n$ and $\bar u_n$ as:
\begin{align}
    Q &= \sum_{k > 0, \beta}{u_n^k (u_n \bar u_n)^\beta F_{k,\beta}(\cdots)} + \sum_{\beta}{(u_n \bar u_n)^\beta F_{0,\beta}(\cdots)} + \sum_{k < 0, \beta}{\bar u_n^{-k} (u_n \bar u_n)^\beta F_{k,\beta}(\cdots)}\nn \\
    &= \sum_{k > 0}{u_n^k G_k(u_n \bar u_n, \cdots)} + G_0(u_n \bar u_n, \cdots) + \sum_{k < 0}{\bar u_n^{-k} G_k(u_n \bar u_n, \cdots)},\;\;\;G_k(u_n \bar u_n, \cdots) \defn \sum_\beta{(u_n \bar u_n)^\beta F_{k,\beta}(\cdots)},   
\label{eq:Qdeltaexpansion}
\end{align}
where the first line is just a Taylor expansion with $F_{k,\beta}$ defined such that it does not depend on $u_n$ and $\bar u_n$ and $\delta(F_{k,\beta}) = -k$ since $Q$ is balanced, and in the second line $G_k$ depends on the combination $u_n \bar u_n$ and the rest of the variables and $\delta(G_k) = -k$.
Imposing the condition of Eq.~(\ref{eq:Qcondition}) gives
\begin{equation}
    \sum_{k > 0}{[u_n^k - (2b_n - u_n)^k] G_k(u_n (2b_n - u_n), \cdots)} + \sum_{k < 0}{[(2b_n - u_n)^{-k} - u_n^{-k}] G_k(u_n (2b_n - u_n), \cdots)} = 0.
\label{eq:Qdeltaclear}
\end{equation}
Now note that the terms corresponding to different $k$ are linearly independent, since they have different $\delta$-degrees with respect to the remaining variables in $\cdots$, which are $v_n$, $\bar{v}_n$, $\{z_{j \neq n, n+1}\}$, $\{\bar{z}_{j \neq n, n+1}\}$.  
In particular, we can define a $\delta'$-degree by setting $\delta'(u_n) = \delta'(\bar{u}_n) = 0$, $\delta'(v_n) = \delta'(z_{j \neq n, n+1}) = 1$, and $\delta'(\bar{v}_n) = \delta'(\bar{z}_{j \neq n, n+1}) = -1$, and from the definitions, we get that $\delta'(G_k) = \delta'(F_{k,\beta}) = \delta(F_{k,\beta}) = -k$, which means that $\delta'([u_n^k - (2b_n - u_n)^k]G_{k}) = \delta'([(2b_n - u_n)^k - u_n^k]G_{k})  = -k$.
Hence from Eq.~(\ref{eq:Qdeltaclear}) we get that
\begin{align}
    &u^k_n G_k(u_n (2b_n - u_n), \cdots) = (2b_n - u_n)^{k} G_k(u_n (2b_n - u_n), \cdots)\;\;\text{if}\;\;k > 0, \nn \\
    &(2b_n - u_n)^{-k} G_k(u_n (2b_n - u_n), \cdots) = u^{-k}_n G_k(u_n (2b_n - u_n))\;\;\text{if}\;\;k < 0. 
\end{align}
Given that $b_n \neq 0$, this immediately means that $G_k(u_n (2b_n - u_n)) = 0$ for all $u_n$ if $k \neq 0$, which, since $G_k$ is analytic, also means that $G_k = 0$ if $k \neq 0$. 
Using Eq.~(\ref{eq:Qdeltaexpansion}), we have that
\begin{equation}
    Q = G_0(u_n \bar u_n, \cdots)\;\;\;\implies\;\;\;Q = \mathcal{P}_n[Q], 
\end{equation}
where we have used the property of Eq.~(\ref{eq:Pnaction}).
This shows the momentum exchange symmetry of $Q$, completing the proof. 
\end{proof}

We now use Lemmas \ref{lem:a=0} and \ref{lem:collisions}, and Theorem \ref{thm:Pexchange} to show the following theorem.
\begin{theorem}
\label{thm:collisions}
   {When $a_n \neq 0$ for at least a single $n$, any $Q \in \mathbb{C}[\{z_i \bar z_j\}]$ that satisfies Eq.~(\ref{eq:collisionconserved}) is also symmetric under arbitrary permutations of $N$ momenta. Hence, it is functionally dependent on $E$ and $E_{\rm cm}$, i.e., $Q \in \mathbb{C}[E, E_{\rm cm}]$.} 
\end{theorem}
\begin{proof}
{Given any two rods $n$ and $n+1$, there are just two possibilities of their lengths: (i) $a_n \neq 0$ or $a_{n+1} \neq 0$, or (ii) $a_n = a_{n+1} = 0$. 
In case (i), we can apply Lemma \ref{lem:collisions}, and obtain that $Q$ is invariant under momentum exchange of rods $n$ and $n+1$.
Similar to arguments below Eq.~(\ref{eq:U1transform}), one can show that this momentum exchange invariance along with the $U(1)$ invariance imposed by free motion, also implies invariance under position exchange of rods $n$ and $n+1$. 
Naturally this also implies invariance under the simultaneous exchange of momenta and positions of those rods, which is a smaller group, and we will exploit this in the proof below. 
In case (ii), we can apply Lemma \ref{lem:a=0}, and obtain that $Q$ is invariant only under the simultaneous exchange of their momenta and positions of the rods $n$ and $n+1$, not necessarily under the individual exchanges of momenta or the positions. 
Hence, applying to all pairs of rods, we obtain that $Q$ should at least be invariant under joint permutation group of positions and momenta, which we {earlier} denote{d} as the group $S^{(x,p)}_N$.
To make the notation clear, we are considering the following four groups, which we will denote by the operations
\begin{gather}
    S^{(x)}_N \times S^{(p)}_N = \{(\sigma, \tau),\;\;\;\sigma,\tau \in S_N\}\nn \\
    S_N^{(x)} = \{(\sigma, 1),\;\;\sigma \in S_N\},\;\;\;S^{(p)}_N = \{(1, \sigma),\;\;\sigma \in S_N\},\;\;\;S^{(x,p)}_N = \{(\sigma, \sigma),\;\;\;\sigma \in S_N\}. 
\label{eq:groupops}
\end{gather}
Note that $S_N^{(x)}$, $S_N^{(p)}$, and $S_N^{(x,p)}$ are all distinct $S_N$ subgroups of the larger group $S_N^{(x)} \times S^{(p)}_N$. 
If all the rod lengths are non-zero, we immediately have that $Q$ is invariant under the full group $S^{(p)}_N$, which, along with $U(1)$ invariance of free motion, by Theorem \ref{thm:Pexchange} implies that $Q \in \mathbb{C}[E, E_{\rm cm}]$. 
On the other hand, if all the rod lengths are zero, we only have that $Q$ is invariant under $S_N^{(x,p)}$, and by Theorem \ref{thm:a=0} implies that there is a larger set of conserved quantities independent of $E$ and $E_{\rm cm}$. 
Here we are interested in the case where at least a single $a_n \neq 0$. 
In this case, in addition to invariance under $S^{(x,p)}_N$, it should be invariant under at least a single transposition $(1, \tau) \in S^{(p)}_N$. 
As a consequence of $U(1)$ conservation and {Eq.~(\ref{eq:SpSxinvariance})}, it should also be invariant under the transposition $(\tau, 1) \in S^{(x)}_N$. 
Hence the full group that $Q$ is invariant under is 
\begin{equation}
    G = \langle (1, \tau), (\tau, 1), S^{(x,p)}_N \rangle, 
\end{equation}
where $\langle \cdots \rangle$ denotes the group generated by elements $\cdots$. 
We then have that
\begin{equation}
    (\rho, \rho)(1, \tau)(\rho, \rho)^{-1} = (1, \rho \tau \rho^{-1}) \in G\;\;\;\text{and}\;\;\;(\rho, \rho)(\tau, 1)(\rho, \rho)^{-1} = (\rho \tau \rho^{-1}, 1) \in G\;\;\;\forall\;\;\rho \in S_N.  
\end{equation}
Since $\{\rho \tau \rho^{-1}\}$ includes all possible transpositions, which in turn generate all permutations in $S_N$, we obtain that $(1, \sigma) \in G$ and $(\sigma, 1) \in G$ for all $\sigma \in S_N$, hence $G = S^{(x)}_N \times S^{(p)}_N$. 
This means that $Q$ should then be a conserved quantity of the SMED dynamics studied in Sec.~\ref{sec:SMED}, and then we can apply Theorem \ref{thm:Pexchange} to conclude that $Q \in \mathbb{C}[E, E_{\rm cm}]$.
This rules out the existence of additional independent conserved quantities in systems when at least one rod has non-zero length, which includes the case when all rods have equal non-zero lengths, such as those studied in Ref.~\cite{bagchi}.}

\end{proof}

\section{Discussion and Outlook}
\label{sec:discussion}
{In this work we investigated the question of whether the anomalous dynamical behavior observed in the system of harmonically confined hard-rods originates from a hidden conservation law.
By showing that any analytic conserved quantity in the phase space variables $\{x_i\}$ and $\{p_i\}$ should be invariant under a combination of a continuous $U(1)$ symmetry and a momentum permutation symmetry $S_N^{(p)}$, we showed that the only independent conserved quantities are the total energy $E$ and the center-of-mass energy $E_{\rm cm}$ whenever the system has at least a single rod of non-zero length. 
This rules out the existence of an independent hidden {analytic} conserved quantity {in that system}.
{Along the way, we obtained exhaustive statements on conserved quantities in a variety of related systems, such as non-interacting rods, point particles, and Stochastic Momentum Exchange Dynamics, results of which are summarized in Tab.~\ref{tab:symmetry-summary}.}
{It would be interesting in the future to extend these results to also explore or rule out potentially non-analytic conserved quantities.}
{In many ways, this absence of a hidden analytic conserved quantity deepens the mystery of harmonically confined rods, and below we comment on various aspects of this puzzle and potential future directions. }
\subsection{{Poincare sections and Lyapunov exponents}}
\begin{figure}[hbt!]
\centering

\begin{subfigure}{.35\textwidth}
\centering
\includegraphics[width=1.0\linewidth]{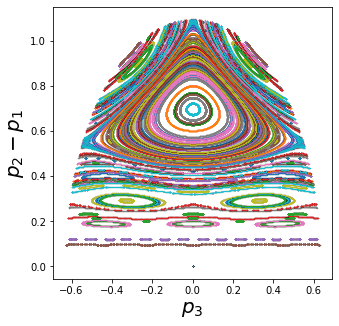}
\caption{}
\end{subfigure}%
\begin{subfigure}{.35\textwidth}
\centering
\includegraphics[width=1.0\linewidth]{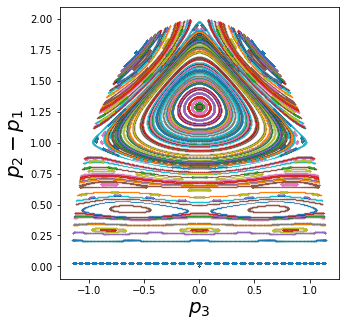}
\caption{}
\end{subfigure}%
\begin{subfigure}{.35\textwidth}
\centering
\includegraphics[width=1.0\linewidth]{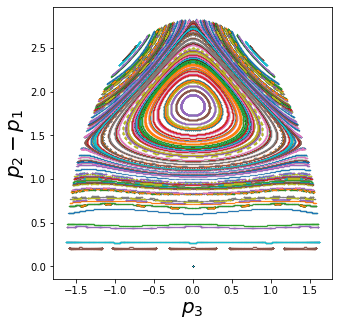}
\caption{}
\end{subfigure}
\caption{{\bf Poincar\'e section for equal length hard rods}: (a) $E=2.6$, (b) $E=4$, (c) $E=6$. We have taken $100$ different initial conditions used $E_{\rm cm} =0$ with rod lengths $a = 1$ in each of the plots.
\label{poincareequallength}}
\end{figure}
{We should however remark} this result is not completely unexpected, and there have been earlier numerical results in the system of equal length rods that suggest this absence.
Ref.~\cite{bulchandani} commented that the fractal structure in the Poincaré section of such systems{, also reproduced in Fig.~\ref{poincareequallength},} is not compatible with the existence of a hidden analytic conserved quantity.
Further, even though numerous regular orbits with zero Lyapunov exponents have been observed for this system~\cite{bagchi} {for $N = 3$}, which initially {was one of the} motivations for the search for a hidden conserved quantity, a closer look does suggest the presence of rare chaotic trajectories showing scattered Poincaré section {even in that case}.
{Concretely, the Poincar\'e section for three rods is constructed by fixing the center-of-mass energy to $E_{\rm cm}=0$, which is equivalent to imposing $\sum_{i=1}^3 x_i = 0$ and $\sum_{i=1}^3 p_i = 0$, together with the conservation of the total energy $E$.}
One then records two independent linear combinations of the phase-space variables $\{x_1,x_2,x_3,p_1,p_2,p_3\}$ each time the trajectory crosses the collision surface $x_2-x_1=a$, i.e., when rods $1$ and $2$ collide.
Since these conditions {altogether} impose four constraints on the six-dimensional phase space, the resulting Poincar\'e section is two-dimensional.
{The full Poincare section for the equal length rod case is} reproduced in {Fig.~\ref{poincareequallength} for various energies, and} Fig.~\ref{chaotictrajectory}(a){shows it} for a single special initial condition, which shows scattered points upon zooming in Fig.~\ref{chaotictrajectory}(b).
The same initial condition {also} has a nonzero Lyapunov exponent as shown in Fig.~\ref{chaotictrajectory}(c).
The Lyapunov exponents are computed using the method described in \cite{benettin}. 
We begin with two initial conditions separated by a distance $\epsilon$ and evolve both for a time $\tau$, after which the distance becomes $d_1$.
The second trajectory is then shifted along the line of separation to restore the distance to $\epsilon$, and this process is repeated for $M$ steps, producing the sequence $d_1,d_2,\ldots, d_M$.
The Lyapunov exponent $\lambda(t)$ is defined as:
\begin{equation}
    \lambda(t)=\frac{1}{\lfloor\frac{t}{\tau}\rfloor}\sum_{i=1}^{\lfloor\frac{t}{\tau}\rfloor}\frac{1}{\tau}\ln{\frac{d_i}{\epsilon}}.
\end{equation}
where $\lfloor x \rfloor$ is the greatest integer less than or equal to $x$.
For the chaotic initial condition, $\lambda(t)$ vs $t$ is plotted in Fig.~\ref{chaotictrajectory}(c) for different values of $\epsilon$, and the long time values of $\lambda(t)$ can be seen to converge to a nonzero value as $\epsilon$ goes to zero.
In contrast, for a randomly chosen initial condition {with $N = 3$}, we see that $\lambda(t)$ is approximately zero.
{Note that this phenomenology changes when} $N>3$, where Lyapunov exponents {were found to be} non-zero for most of the initial conditions in~\cite{bagchi}.
{Nevertheless,} the system {there too} does not thermalize {as expected} and shows non-ergodicity.
\begin{figure}[t!]
\centering
\begin{subfigure}{.35\textwidth}
\centering
\includegraphics[width=1.0\linewidth]{{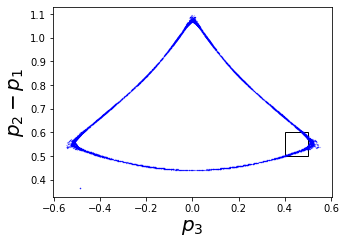}}
\caption{}
\end{subfigure}%
\begin{subfigure}{.35\textwidth}
\centering
\includegraphics[width=1.0\linewidth]{{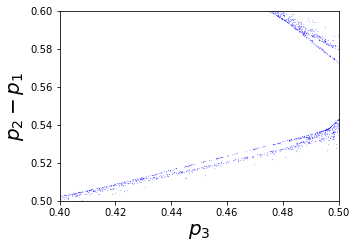}}
\caption{}
\end{subfigure}%
\begin{subfigure}{.35\textwidth}
\centering
\includegraphics[width=1.0\linewidth]{{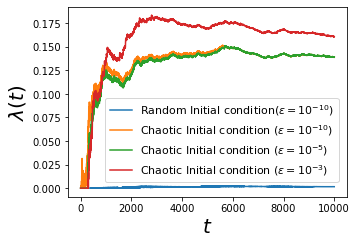}}
\caption{}
\end{subfigure}
\caption{{\bf Chaotic trajectory}: (a) Poincare section for a single initial condition $(x_1, x_2, x_3, p_1, p_2, p_3) \approx (-1.002, 0.001, 1.001, 0, -0.545, 0.545)$
showing scattered points, {as evident} when zoomed-in {on the square region as shown} in (b).
This is shown for a system of three rods of equal length $a_1 = a_2 = a_3 = 1$, and the system has total energy $E=2.6$, and {center-of-mass energy} $E_{\rm cm}=0$. (c) Time dependent Lyapunov exponent $\lambda(t)$ for the chaotic trajectory in panels (a,b) using perturbations $\epsilon=10^{-10}$ (orange), $10^{-5}$(green), $10^{-3}$(red), together with a randomly selected initial condition with $\epsilon=10^{-10}$(blue) ($\tau=1$ for all four cases). The random trajectory remains near zero, whereas the chaotic trajectory approaches a nonzero value.\label{chaotictrajectory}}
\end{figure}
{This $N = 3$ behavior as seen in Fig.~\ref{poincareequallength} might} suggest that the phase space {is} composed of regular regions {even in the absence of additional analytic conserved quantities}, {making it} reminiscent of the behavior seen in weakly perturbed integrable systems, such as the Chirikov standard map~\cite{CHIRIKOV1979263}.
In that setting, the KAM theorem~\cite{Arnold2009} implies that many invariant tori---and the associated action variables---persist for sufficiently small perturbations on a set of large measure, although not as a complete global set of exact conserved quantities across the whole phase space.
By analogy, one may wonder whether the hard-rod system admits similarly ``dressed'' or quasi-conserved quantities that are only well defined on certain regions of phase space rather than globally, {and whether such quantities can arise as deformations} of conserved quantities in one of the two natural integrable limits --{one} when the harmonic potential is completely turned off, {and another} the rod lengths are all set to zero.
{It would be interesting to explore this possibility in future work, and whether such approximate integrals of motion can also explain the observed behaviors for $N > 3$~\cite{bagchi}.}
\subsection{{Speciality of equal length rods}}
An additional question concerns the apparent dynamical role of the rod lengths.
Our analytical result treats all systems with at least one non-zero rod on the same footing: whether the non-zero lengths are equal or unequal, the collision constraints rule out any independent globally analytic conserved quantity beyond $E$ and $E_{\rm cm}$.
Thus, equal rod lengths are not exceptional at the level of exact analytic conservation laws.
Nevertheless, numerical results {do} suggest that the equal-length case {are} dynamically distinct, in the sense that it appears to support a comparatively large prevalence of regular trajectories, whereas unequal {length rods appear to} produce more visibly scattered sections in some parameter regimes.
The Poincar\'e section for equal-length rods where all the length of all the rods are unity is shown in Fig.~\ref{poincareequallength} for various initial conditions for different values of energy.
{We then also explore the} Poincar\'e sections for systems of rods with unequal length {rods}, which has not been done in previous works.
For $3$ rods of unequal lengths $a_1,a_2,a_3$, there are only two lengths which enter the dynamics viz. $b_1 = \frac{a_1 + a_2}{2}$ (during collisions of rods $1$ and $2$) and $b_2 = \frac{a_2 + a_3}{2}$ (during collisions of rods $2$ and $3$).
The corresponding Poincar\'e sections are shown in Fig.~\ref{poincareunequallength}.
We see that some values of the parameters $b_1,b_2,E$ are evidently less regular with lots of scattered points whereas some seem as regular as the equal rod case.
In particular, we observe that the prevalance of regular orbits increase when the ratio $b_1/b_2$ is close to unity, which is the ratio in the case of equal length rods.
This suggests that the equal length rods case is special {in their dynamics, and it might be worthwhile in the future to nail down the precise distinction of this case further and develop proofs that are not agnostic to the rod lengths.
It would also be an interesting numerical exercise to check if unequal length rods thermalize as expected for $N > 3$.
} 
\begin{figure}[t]
\centering
\begin{subfigure}{.35\textwidth}
\centering
\includegraphics[width=1.0\linewidth]{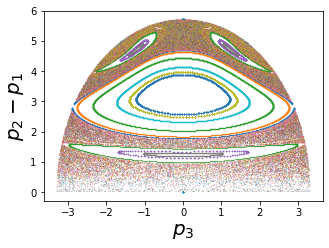}
\caption{}
\end{subfigure}%
\begin{subfigure}{.35\textwidth}
\centering
\includegraphics[width=1.0\linewidth]{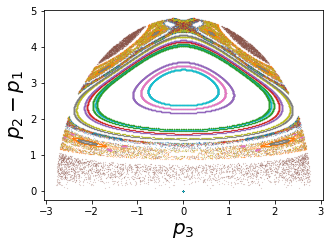}
\caption{}
\end{subfigure}%
\begin{subfigure}{.35\textwidth}
\centering
\includegraphics[width=1.0\linewidth]{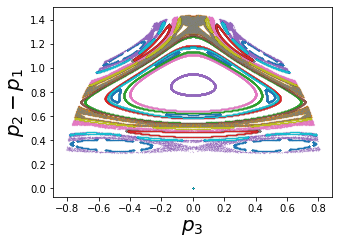}
\caption{}
\end{subfigure}
\caption{{\bf Poincar\'e section for unequal length hard rods}: (a) $b_1=1$, $b_2=3.2$, $E=13$, $E_{\rm cm} = 0$ (b) $b_1=0.5$, $b_2=1.6$, $E=7$, $E_{\rm cm} = 0$, (c) $b_1=1$, $b_2=1.5$, $E=2.1$, $E_{\rm cm} = 0$. {The regularity of these Poincar\'e sections appears to increase when the ratios of the rod lengths approaches unity.}\label{poincareunequallength}}
\end{figure}
\subsection{{Beyond harmonically confined rods}}
{ First, it would also be interesting to extend this proof to hard rods in other potentials. However, note that}
{
the present analysis relies importantly on the harmonic nature of the confining potential. In particular, the free dynamics reduces to independent rotations in the $(x_i,p_i)$ planes, giving a $U(1)$ symmetry whose invariant algebra can be characterized explicitly. This makes the first step of our argument especially tractable. For a general confining potential, this simple rotation structure is absent, and an exhaustive characterization of the invariants of the free motion would have to be replaced by an appropriate analysis of the corresponding flow.

A second ingredient of our approach is more structural: collision rules impose symmetry constraints on conserved quantities. In the present hard-rod problem these constraints lead to momentum and position permutation symmetries that, when combined with the free-motion invariants, are strong enough to rule out additional conserved quantities in the class considered here. Thus, although our theorem is specific to the harmonic setting, the general strategy of first characterizing free-motion invariants and then imposing collision-induced symmetry constraints may be applicable more broadly whenever the corresponding structures can be described algebraically.}

Finally, a natural step is to apply the similar algebraic methods to {other classical many-body} systems {that show regular regions without any integrals of motion.}
{A recent set of examples are systems} with dipole or higher multipole {moment} conservation, including classical fractons and phase-space fractons.
Ergodicity breaking {there is said to} arise from kinematic constraints rather than hidden exact integrals {of motion}~\cite{ PhysRevB.109.054313,PhysRevB.110.024305,b974-mpkc}, {and it would be interesting to systematically obtain or rule out analytic conserved quantities}.
More generally, our work {motivates the development of a systematic framework for an algebraic understanding of conserved quantities in classical many-body systems, analogous to systematic constructions such as commutant algebras that exist for quantum many-body systems~\cite{PhysRevX.12.011050, PhysRevX.14.041069, PhysRevB.107.224312, moudgalya2022from}.
Such a framework, which might potentially exist in the literature in some other form, can be used as a powerful tool to constrain the kinds of conserved quantities that can occur in classical many-body systems, and might also unravel novel kinds of conserved quantities that lead to novel phenomena.}
\section*{Acknowledgements}
We acknowledge the use of ChatGPT-5.5 for the generation of many of the proof ideas presented here, which were {then} independently verified, simplified, and heavily streamlined} by the authors.
S.M. acknowledges support from the Munich Center for Quantum Science and Technology (MCQST), which is supported by the Deutsche Forschungsgemeinschaft (DFG, German Research Foundation) under Germany’s Excellence Strategy--EXC--2111--390814868. AD acknowledges the J.C. Bose
Fellowship (JCB/2022/000014) of the Science and Engineering Research Board of the Department of Science and Technology, Government of India and the support of the DAE, Government of India, under projects nos. 12-R\&D-TFR-5.10-1100 and RTI4001.

\bibliographystyle{unsrturl}
\bibliography{references}
\end{document}